\documentclass[preprints,article,accept,moreauthors,pdftex]{mdpi} 

\firstpage{1} 
\pubvolume{1}
\issuenum{1}
\articlenumber{0}
\pubyear{2021}
\copyrightyear{2020}

\history{}
\pdfoutput=1

\definecolor{greenlink}{rgb}{0.3, 0.6, 0.3}
\usepackage{subcaption}

\usepackage{mathtools,amssymb}

\usepackage{braket}

\newtheorem{thm}{Theorem}

\Title{Entropy of quantum states}

\Author{Paolo Facchi $^{1,2}$\orcidA{}, Giovanni Gramegna $^{3}$\orcidB{} and Arturo Konderak $^{1,2,}$*\orcidC{}}

\AuthorNames{Paolo Facchi,  Giovanni Gramegna and Arturo Konderak}

\address{%
$^{1}$ \quad Dipartimento di Fisica and MECENAS, Universit\`a di Bari, I-70126 Bari, Italy\\
$^{2}$ \quad INFN, Sezione di Bari, I-70126 Bari, Italy\\
$^{3}$ \quad Dipartimento di Fisica, Universit\`a di Trieste, I-34151 Trieste, Italy
}

\corres{Correspondence: arturo.konderak@ba.infn.it}

\abstract{Given the algebra of observables of a quantum system subject to selection rules, a state can be represented by different density matrices. As a result, different von Neumann entropies can be associated with the same state.
Motivated by a minimality property of the von Neumann entropy of a density matrix with respect to its possible decompositions into pure states, we give a purely algebraic definition of entropy for states of an algebra of observables, thus solving the above ambiguity.
 The entropy so defined satisfies all the desirable thermodynamic properties, and reduces to the von Neumann entropy in the quantum mechanical case. Moreover, it can be shown to be equal to the von Neumann entropy of the unique representative density matrix belonging to the operator algebra of a multiplicity-free Hilbert-space representation.}


\begin{document}

\section{Introduction}

In 1931, von Neumann~\cite{von1996mathematical} found a connection between two branches of physics: quantum mechanics and thermodynamics. If a system satisfies the laws of thermodynamics, its entropy is well defined. With this in mind, von Neumann obtained that a quantum system, described by a density matrix $\rho$, has entropy
\begin{equation}\label{vnentropy}
\mathcal{S}_{\mathrm{VN}}(\rho)=-k_{\mathrm B}\operatorname{Tr} (\rho \log \rho),
\end{equation}
where $k_{\mathrm B}$ is the Boltzmann constant.

Besides its importance from a fundamental point of view, von Neumann entropy is useful also to answer practical questions in quantum information theory, for example when dealing with multipartite systems and one wants to characterize the entanglement between them: in this context, it has been shown~\cite{popescu97,donald2002uniqueness} that a particularly meaningful measure of the entanglement contained in a pure state shared by two parties is the von Neumann entropy of the reduced state of one party, since it allows to characterize the usefulness of such entanglement in the thermodynamic limit when multiple copies of the state are available. This measure is then extended to mixed states by exploiting the convex structure of the set of quantum states, and taking the infimum over all the possible decompositions into pure states~\cite{horodecki2001entanglement}.

As quantum field theory developed, attempts to extend equation~(\ref{vnentropy}) to a broader scheme have been made. Von Neumann entropy is used to evaluate the entropy of a black hole~\cite{bombelli1986quantum,holzhey1994geometric,srednicki1993entropy}, which originates from the lack of knowledge of the system inside of it. However, ``it is never hard to find trouble in field theory''~\cite[p.~74]{bjorken1966relativistic}, as ambiguities in this definition arise from the dependence on the cutoffs introduced to regularize the theory. In~\cite{balachandran2013quantum,sorkin2014expressing} this ambiguity was traced back to the ambiguity in the definition of a density matrix associated with a state in an algebraic theory. In fact, the proper mathematical formalization of a quantum field theory requires the introduction of $C^*$-algebras~\cite{haag1964algebraic}. In this context, in general the set of observables is not the full operator algebra, but a subalgebra~\cite{bratteli2012operator,davidson1996c}. 

We can get a glance of the ambiguity by the following example.
Consider the algebra of diagonal $n\times n$ matrices
\begin{equation}
\mathfrak{A}=\{A\in M_n: A_{ij}=0 \textnormal{ for } i\neq j \}.
\end{equation}
For any density matrix $\rho$ the result of a measurement is
\begin{equation}
\operatorname{Tr}(\rho A)=\sum_{i}\rho_{ii} A_{ii},
\end{equation}
and depends only on the diagonal elements of $\rho$. Thus, density matrices define the same state as long as they have the same diagonal elements. However, their von Neumann entropy, as defined in equation~(\ref{vnentropy}), can be different.
Which density matrix is associated with the correct physical entropy?

In order to attain an unambiguous definition of entropy, it is necessary to study states as abstract entities rather than density matrices.
In this abstraction, the only relevant feature of the set of quantum states is its convex structure, and the problem translates into the more general question of giving a sensible definition of entropy for points in a convex set. This problem has been studied in~\cite{uhlmann1970shannon,balachandran2013entropy}.

The purpose of this Article is to give an unambiguous definition of entropy for a state over an algebra of observables, connecting this problem to the definition of entropy for points on a convex set. Also, the physical implications of this mathematical definition will be investigated, together with its thermodynamic interpretation and its connection with von Neumann entropy. The study will be carried out for a finite dimensional algebra.

The Article is organized as follows. In section~\ref{sec:ObservablesStates} we introduce the essential notation and briefly recall the algebraic approach to quantum theory. In particular, we discuss a structure theorem for finite dimensional $C^*$-algebras, which plays a central role in the derivation of the results presented. Then, in section~\ref{sec:entropies} we briefly discuss the relation between von Neumann entropy of a density matrix in  quantum mechanics and the Shannon entropy of its possible  decompositions into pure states, which motivates  the  definition of the entropy for a state over a $C^*$-algebra as the infimum over its possible decompositions. In section~\ref{chapentrcstaralgeb} we explicitly  compute the quantum entropy of a state by using first a generic faithful representation, and then the GNS construction, and we show its connection with von Neumann entropy. We also discuss some physical implications by extending a thermodynamic argument due to von Neumann to the algebraic setting. Finally, in section~\ref{sec:conc} we conclude the paper with some remarks.

\section{Algebraic approach: observables and states}\label{sec:ObservablesStates}
The formal description of a quantum field theory is given in terms of algebras~\cite{haag1964algebraic,araki1999mathematical}. The main idea is to define observables for each region of space-time, such that observables associated with casually disjointed regions are compatible (or simultaneously measurable).

The set of observables $\mathfrak A$ is required to satisfy certain properties, that define the structure of a $C^*$-algebra. One considers the observables of a given experiment, and defines states as positive linear functionals giving the expectation values of the measurement outcomes. This is at variance with the standard quantum mechanics description on Hilbert spaces, where one starts by considering the set of vector states, and then defines the observables as operators on this set.

A $C^*$-algebra is a Banach space (i.e.\ a normed and complete vector space) $\mathfrak{A}$ with a product
\begin{equation}\label{product}
(A,B)\in \mathfrak{A}\times\mathfrak{A} \mapsto AB \in\mathfrak{A}
\end{equation}
and an involution
\begin{equation}\label{involution}
A\in \mathfrak{A} \mapsto A^* \in\mathfrak{A},
\end{equation}
satisfying $\left\lVert{A}\right\rVert^2=\left\lVert{A^*A}\right\rVert$.

An algebra can be \textit{represented} as an algebra of operators on a Hilbert space $\mathcal{H}$. More precisely, a representation of the theory is a pair $(\mathcal{H},\pi)$ where $\pi$ is a linear map from $\mathfrak A$ to $\mathcal{B}(\mathcal{H})$ preserving~(\ref{product}) and~(\ref{involution}) and $\mathcal{B}(\mathcal{H})$ is the algebra of bounded operators on $\mathcal{H}$. A representation is said to be faithful when $\pi(A)=0$ if and only if $A=0$.

Given an algebra of observables, a state is characterized by the measurement outcomes. States are defined as functionals
\begin{equation}
\omega:\mathfrak{A}\rightarrow \mathbb{C}
\end{equation}
satisfying
\begin{enumerate}
\item[(a)] $\omega(A^*A)\geqslant 0$,
\item[(b)] $\omega(\mathbb{I})=\left\lVert{\omega}\right\rVert=1$.
\end{enumerate}
Here $\mathbb{I}$ is the unit element of the algebra. The definition can be extended to non-unital algebras, see~\cite{bratteli2012operator}.
The convex combination of two states $\omega_1$ and $\omega_2$,
\begin{equation}
\label{eq:convcomb}
\omega = \lambda \omega_1 + (1-\lambda) \omega_2,
\end{equation}
with $\lambda\in (0,1)$, is still a state.
A state is called \emph{pure} or \emph{extremal} if it cannot be written as a convex combination of other states, that is if equation~\eqref{eq:convcomb} implies that $\omega_1=\omega_2=\omega$.
The states over an algebra $\mathfrak{A}$ with a unit element form a convex weakly-* compact set and coincide with the weak-* closure of the convex envelope of its pure states. In other words, we can always decompose a state into pure states.

In the standard quantum mechanical approach states are represented by density matrices $\rho$, and the expectation value of an observable $A\in \mathcal{B(H)}$ is given by
\begin{equation}
\omega_\rho (A) 
=\operatorname{Tr}(\rho A),
\end{equation}
which becomes $\langle{\psi}|{A\psi}\rangle$ for a vector state, that is a rank-1 projection $\rho=\ket{\psi}\!\bra{\psi}$, with $\|\psi\|=1$. It is immediate to verify that this is a functional  satisfying both properties  (a) and (b), and thus is a state over the full  operator  algebra $\mathcal{B(H)}$.

In fact, one can prove that in the algebraic description a state can be always realized in this way, using the GNS construction~\cite{gelfand1994imbedding,segal1947irreducible}. Given  a $C^*$-algebra $\mathfrak{A}$ and a state $\omega$, there exists (up to a unitary transformation) a unique representation $(\mathcal{H}_\omega,\pi_\omega)$ and a unique unit vector $\Omega_\omega\in\mathcal{H}_\omega$   such that
\begin{equation}
\omega(A)=\langle{\Omega_\omega}|{\pi_\omega(A)\Omega_\omega}\rangle.
\end{equation}
Notice, however, that, at variance with quantum mechanics, in general the algebra $\pi_\omega(\mathfrak{A})$ is \emph{smaller} than the full operator algebra $\mathcal{B}(\mathcal{H}_\omega)$, and a vector state (and in particular $\ket{\Omega_\omega}\!\bra{\Omega_\omega}$) does \emph{not} necessarily correspond to a pure state and vice versa. This is the case when the quantum system is subject to superselection rules, or is composed by more than one thermodynamic phase~\cite{araki1999mathematical,strocchi2008introduction}. 

In this Article we are going to deal only with finite dimensional $C^*$-algebras. In this case the algebra is isomorphic to the direct sum of full matrix algebras~\cite{davidson1996c}:
\begin{thm}[Structure theorem]\label{structuttheorzero}
	Every finite dimensional $C^*$-algebra can be  faithfully represented as the direct sum of full matrix algebras 
	\begin{equation}\label{structuttheoruno}
	\pi(\mathfrak{A})= M_{n_1}\oplus M_{n_2}\oplus \dots \oplus M_{n_k},
	\end{equation}
	and thus any finite dimensional $C^*$-algebra is unital. Moreover, any faithful non-degenerate finite dimensional  representation has the form
	\begin{equation}\label{structtheordue}
	\pi(\mathfrak{A})=M_{n_1}^{(m_1)}\oplus M_{n_2}^{(m_2)}\oplus \dots \oplus M^{(m_k)}_{n_k},
	\end{equation}
	up to a unitary transformation, with
	\begin{equation}\label{structtheortre}
	M_{n}^{(m)}=\big\{ \underbrace{X\oplus X \oplus \dots\oplus X}_\text{m}:X\in M_n\big\}.
	\end{equation}
	the algebra obtained by repeating $m$ times the same element of $M_n$.
\end{thm}
As a result of Theorem~\ref{structuttheorzero}, any finite dimensional algebra can be faithfully represented as a finite dimensional algebra of operators, as in equations~(\ref{structtheordue}) and~(\ref{structtheortre}). Note that standard quantum mechanics corresponds to the case $k=m_1=1$.

\section{Shannon entropy and von Neumann entropy}\label{sec:entropies}
Given a probability vector $\vec p=(p_1,p_2,\dots,p_n)$, with $p_i\geqslant 0$ and $\sum_i p_i =1$,
 its Shannon entropy is defined as
\begin{equation}
\mathrm{H}(\vec p)= -
\sum_{i=1}^n p_i\log p_i.
\label{eq:shannondef}
\end{equation}
As required for the entropy, $\mathrm{H}$ is a strictly concave function, that is
\begin{equation}
\mathrm{H}(\lambda \vec{p}+(1-\lambda)\vec{q})\geqslant \lambda \mathrm{H}( \vec{p})+(1-\lambda)\mathrm{H}(\vec{q}),
\end{equation}
for all $\lambda\in(0,1)$, with equality holding if and only if $\vec p=\vec q$.

There is a connection between Shannon entropy and von Neumann entropy in quantum mechanics. For a given density matrix $\rho$ with eigenvalues $\lambda_i$, its von Neumann entropy is (by setting the Boltzmann constant $k_{\mathrm B}=1$)
\begin{equation}\label{vnentropyformula}
\mathcal{S}_{\mathrm{VN}}(\rho)\equiv- 
\operatorname{Tr}{\rho \log \rho}=- 
\sum_{i=1}^N \lambda_i\log \lambda_i=\mathrm{H}(\vec\lambda).
\end{equation}
that is the Shannon entropy of its eigenvalues. It can be seen that von Neumann entropy is also strictly concave as a function of $\rho$.

There is a deeper connection between the two entropies. Given a state in quantum mechanics, described by a density matrix $\rho$, it can always be seen as a probabilistic mixture of vector states
\begin{equation}\label{decomp}
\rho = \sum_{i=1}^n p_i |{\phi_i}\rangle\langle{\phi_i}|,\quad p_i\geqslant 0,\quad \sum_{i=1}^n p_i=1.
\end{equation}
However, this decomposition is not unique and the same state can be prepared in different ways as a convex combination of vector states: using the language of convex geometry, the set of quantum states is not a simplex~\cite{Peres2002,bengtsson2017geometry}.

The ambiguity in the preparation of a state is one of the greatest difference between classical and quantum information theory~\cite{preskill2}. For any decomposition~(\ref{decomp}), it is possible to define a Shannon entropy $\mathrm H{(\vec p)}$, and the ambiguity in the preparation is reflected in an ambiguity in the Shannon entropy.
In particular, the von Neumann entropy~(\ref{vnentropyformula}) is the Shannon entropy associated with the spectral decomposition of $\rho$.

The problem of the ambiguity in the ensemble preparation was studied in a seminal paper by Schr\"odinger~\cite{schrodinger1935discussion}, who found a relation between all the preparations of a state, that is  all the possible decompositions of a density matrix into vector states. He proved that for any decomposition~(\ref{decomp}) there exists an $n\times n$ unitary matrix $U$ such that
\begin{equation}\label{shrodingerthm}
p_i=\sum_{j=1}^N \left|{U_{ij}}\right|^2 \lambda_j,
\end{equation}
for all $i=1,\dots,n$,
where $\lambda_j$ are the eigenvalues of the density matrix $\rho$.
Conversely, for any unitary matrix $U$ it is possible to find a decomposition in the form~(\ref{decomp}) such that~\eqref{shrodingerthm} holds.

Notice that, since $U$ is unitary, the matrix $B$ with entries $B_{ij}=\left|{U_{ij}}\right|^2$ is a doubly stochastic matrix, and in particular $\sum_{i=1}^n B_{ij}=1$. Thus the probability vector $\vec p$ is a randomization of the probability vector $\vec \lambda$, through a stochastic process,  namely $\vec p = B \vec \lambda$.

From~(\ref{shrodingerthm}), using the concavity of $h(p)=-p \log p$, we get
\begin{equation}
\mathrm{H}(\vec p)=  \sum_{i} h(p_i )=
\sum_{i} h \Bigl(\sum_j  B_{ij}   \lambda_j\Bigr)
\geqslant
\sum_{i} \sum_j  B_{ij} h(\lambda_j)
=  \sum_j   h(\lambda_j)  = \mathrm{H}(\vec \lambda),
\end{equation}
that is
\begin{equation}\label{shrneum}
\mathrm{H}(\vec p)\geqslant \mathrm{H}(\vec \lambda)=\mathcal{S}_{\mathrm{VN}}(\rho).
\end{equation}

This is a very interesting result, as the von Neumann entropy of a density matrix $\rho$ can be characterized in terms of Shannon entropies of its decompositions into vector states, as the most ordered decomposition, that is the decomposition with the smallest Shannon entropy:
\begin{equation}
\mathcal{S}_{\mathrm{VN}}(\rho) = \inf{ \Big\{\mathrm{H}(\vec{p})\,:\, \rho=\sum_i  p_i |{\phi_i}\rangle\langle{\phi_i}| \Big\}}.
\label{eq:minprop}
\end{equation}
Since convex decomposition into extremal states is a broader concept than orthogonal decomposition, this minimality property suggests a possible definition of entropy for points in a generic convex set, and  in particular for states over a \textit{C*}-algebra. 

\subsection{Majorization relation}

There is a  profound link between the Shannon entropy and the randomness of a probability vector, that sheds light upon equation~\eqref{shrodingerthm}, by giving a partial ordering on the set of probability vectors: the majorization relation~\cite{bhatia2013matrix,marshall1979inequalities}. 
Given two probability vectors $\vec{p}$ and $\vec{q}$ of length $n$, we say that $\vec{p}$ \textit{majorizes} $\vec{q}$ if
\begin{equation}
\sum_{i=1}^k p^{\downarrow}_i\geqslant\sum_{i=1}^k q^{\downarrow}_i, \qquad \forall k=1,\dots n-1,
\end{equation}
and we write
\begin{equation}
\vec{p}\succ\vec{q}.
\end{equation}
Here, $\vec{p}^{\downarrow}$ is the permutation of $\vec p$ such that $p^{\downarrow}_1\geqslant p^{\downarrow}_2\geqslant\dots\geqslant p^{\downarrow}_N$.

The majorization relation is related to the disorder content of a probability vector. For example every probability vector $\vec{p}$ is always in the relation
\begin{equation}
\vec{p}_{\mathrm{det}}\succ \vec{p}\succ\vec{p}_{\mathrm{unif}},
\end{equation}
with respect to the deterministic vector $\vec{p}_{\mathrm{det}}= (1,0,\dots,0)$ and the maximally random probability vector $\vec{p}_{\mathrm{unif}}= (1/n,1/n,\dots,1/n)$. Notice, however, that it can happen that two probability vector $\vec{p}$ and $\vec{q}$ cannot be compared, that is neither $\vec{p}\succ\vec{q}$ nor $\vec{q}\succ\vec{p}$ hold.

Nevertheless, one can prove that $\vec{p}\succ\vec{q}$ if and only if $\vec{q}$ is a randomization of $\vec{p}$, that is $\vec{q}= B \vec{p}$ for some  double stochastic matrix $B$~\cite{hardy,horn1954doubly}. 
Due to the above properties, the majorization relation and its connection with Shannon and von Neumann entropies have proved to play an important role in the quantum resource theories of entanglement~\cite{Nielsen99,Cunden20}
and of quantum coherence~\cite{Winter16,Chitambar16,Cunden21}

In terms of majorization, one can restate Schr\"odinger's theorem~\eqref{shrodingerthm} by saying that the spectral decomposition of a density matrix majorizes all its possible decompositions:
\begin{equation}
\vec{\lambda}\succ \vec{p}.
\end{equation}
Moreover, the Shannon entropy~\eqref{eq:shannondef} is a \textit{Shur concave} function~\cite{marshall1979inequalities,bengtsson2017geometry}, that is if $\vec{\lambda} \succ \vec p$, whence we have
\begin{equation}
\mathrm{H}(\vec{p})\geqslant \mathrm{H}(\vec{\lambda}),
\end{equation}
that is inequality~\eqref{shrneum}. In this sense Shannon entropy is a measure of  disorder.

In the next section, motivated by this minimality property, we will define the entropy of a generic state over a \textit{C*}-algebra as the minimal Shannon entropy over all its possible decompositions into extremal states. By Schr\"odinger's theorem, this quantum entropy will reduce to the von Neumann entropy in the quantum mechanical case.

\section{Entropy of states over a \textit{C*}-algebra}\label{chapentrcstaralgeb}

By mirroring the minimality property~\eqref{eq:minprop}, we now give a definition of entropy for states over an algebra of observables.
Given a finite-dimensional $C^*$-algebra $\mathfrak{A}$, the set of states over $\mathfrak{A}$ is a finite-dimensional convex compact set.
We define the \emph{entropy of a state $\omega$} to be the minimal Shannon entropy among its possible decompositions into pure states, namely
\begin{equation}\label{entropydef}
\mathcal{S}(\omega)=\inf{ \Big\{\mathrm{H}(\vec{p})\, : \, \omega=\sum_i p_i \omega_i, \;\,  \text{$\vec p$ probability vector, \; $\omega_i$ pure states} \Big\}}.
\end{equation}

In the following we will study the properties of this entropy and,  by representing the algebra on a Hilbert space, we will investigate the implications of this formula and its physical interpretations. Different features can be obtained from inequivalent representations of the $C^*$-algebra $\mathfrak{A}$.

Given a representation $(\mathcal{H},\pi)$, it is known that the image $\pi(\mathfrak{A})$ is a $C^*$-subalgebra of the operator algebra $\mathcal{B(H)}$~\cite{bratteli2012operator}. However, we cannot represent any state $\omega$ of the original algebra as a state over $\pi(\mathfrak{A})$. Consider the \textit{representative} state
\begin{align}\label{representationstate}
\omega_\pi:\pi(\mathfrak{A})&\rightarrow \mathbb{C},\\
\pi(A)&\mapsto \omega(A).
\end{align}
This definition makes sense if and only if, for $B\in\mathfrak{A}$:
\begin{equation}\label{condreprstate}
\pi(B)=0 \Rightarrow \omega(B)=0.
\end{equation}
This condition is fulfilled in a \textit{faithful} representation, where by definition $\pi(B)=0$ if and only if $B=0$. Condition \eqref{condreprstate} is also fulfilled in the GNS representation associated with the state $\omega$, where $\pi_{\omega}(A)=0$ implies that $\omega(A)=\braket{\Omega_\omega|\pi_\omega(A)\Omega_\omega}=0$. 
In the following we will compute the entropy~(\ref{entropydef}) using a faithful representation (and later the GNS representation), and will exhibit its connection with the von Neumann entropy of a distinguished representative density matrix in that representation.

\subsection{States over a \textit{C*}-algebra of operators}
In this section, we show that states can be uniquely characterized by density matrices when we deal with a finite dimensional algebra of operators. Moreover, we prove that there exists a unique representative density matrix which is also an element of the algebra. 
\begin{thm}\label{reprstate}
Let $\mathfrak{A}$ be a $C^*$-algebra of operators over a finite-dimensional Hilbert space
\begin{equation}
\mathfrak{A} \subset \mathcal{B}(\mathcal{H}), \qquad \mathrm{dim}\ \mathcal{H}=n< \infty,
\end{equation}
and let $\omega$ be a state over $\mathfrak{A}$. Then, there exists a unique density matrix belonging to the algebra, $\rho_\omega\in \mathfrak{A}$, such that
\begin{equation}\label{statedensitymatrix}
\omega(A) = \operatorname{Tr}(\rho_\omega A), \quad \forall A\in \mathfrak{A}.
\end{equation}
\end{thm}
\begin{proof}
In order to prove the existence of such an element, consider  the Hilbert-Schmidt inner product on $\mathcal{B}\left(\mathcal{H}\right)$,
\begin{equation}\label{hsproduct}
\langle{A}|{B}\rangle_{\mathrm{HS}}\equiv\operatorname{Tr} \big(A^{\dagger}B\big),
\end{equation}
which makes the subspace $\mathfrak{A}$ a Hilbert space.
From Riesz's lemma, for any functional $f\in \mathfrak{A}^*$ there exists a unique $\rho_f \in \mathfrak{A}$ such that:
\begin{equation}
f(A)=\operatorname{Tr}\big(\rho_f^{\dagger}A\big)\qquad \forall A\in \mathfrak{A}.
\end{equation}
In particular, given a state $\omega$ we get a unique operator $\rho_\omega\in \mathfrak{A}$ satisfying $\omega(A) = \operatorname{Tr}(\rho_\omega^\dag A)$ for all $A\in \mathfrak{A}$. 

We now prove that $\rho_\omega$ is a density matrix, that is
$\rho_\omega^{\dagger}=\rho_\omega$,  $\rho_\omega$ is positive, and $\operatorname{Tr}(\rho_\omega)=1$.

 If $B=A^{\dagger}A$ is positive, then
\begin{equation}
\omega(B)=\operatorname{Tr}\big(\rho_\omega^{\dagger} B\big)=\overline{\operatorname{Tr}(B^{\dagger}\rho_\omega)}=\operatorname{Tr}(\rho_\omega B),
\end{equation}
where we used the fact that $\omega(B)$ is real. Since every self-adjoint operator is a linear combination of two positive operators, and every operator is a linear combination of two self-adjoint operators, we have $\operatorname{Tr}(\rho_\omega^{\dagger} A)=\operatorname{Tr}(\rho_\omega A)$ for all $A \in \mathfrak{A}$, whence $\rho_\omega=\rho_\omega^\dagger$.

Since $\rho_\omega$ is self-adjoint, it can be written in its spectral decomposition $\rho_\omega=\sum_i \lambda_i P_i$, with $\lambda_i$ eigenvalues and $P_i$ eigenprojections. Since
\begin{equation}
P_i=\prod_{j \, : \, j\neq i} \frac{\rho_\omega-\lambda_j}{\lambda_i-\lambda_j}.
\end{equation}
we have $P_i \in \mathfrak{A}$ for all $i$. But then:
\begin{equation}
\omega(P_i)=\operatorname{Tr}(\rho_\omega P_i)=\lambda_i \mathrm{dim}\mathcal{H}_i\geqslant 0
\end{equation}
since $P_i=P_i^{\dagger}P_i$ is positive. Here, $\mathcal{H}_i$ is the eigenspace of the eigenvalue $\lambda_i$. Therefore, $\lambda_i\geqslant 0$, and $\rho_\omega$ is positive.

Finally, one has
\begin{equation}
\operatorname{Tr}(\rho_\omega) = \operatorname{Tr}(\rho_\omega \mathbb{I}) = \omega(\mathbb{I})=1.
\end{equation}
Therefore, $\rho_\omega$ is a density matrix.
\end{proof}

For an infinite-dimensional Hilbert space, only  a subclass of states, known as normal states, can be represented by a density matrix. In this setting, equation~\eqref{hsproduct} is not defined for all pairs of bounded operators, and one must recur instead to the duality between bounded operators and trace-class operators~\cite{bratteli2012operator}. 

Observe that, given a state $\omega$, different density matrices can be chosen to represent it. However, $\rho_\omega$ is the only density matrix which is also an element of the algebra $\mathfrak{A}$. So, we have a \textit{distinguished} representative density matrix, and we might think to define the entropy of our system as the von Neumann entropy of this density matrix. A natural question is to understand what is the relation between this von Neumann entropy and the entropy of a state given by formula~(\ref{entropydef}), and in particular whether
\begin{equation}\label{entropyalternative}
\mathcal{S}(\omega)=\mathcal{S}_{\mathrm{VN}}(\rho_\omega)
\end{equation}
holds or not. In the next section, we will study the entropy of a state~(\ref{entropydef}), and we will see that indeed~\eqref{entropyalternative} is true for a faithful and multiplicity-free representation. 

\subsection{Evaluation in a \textit{faithful} representation}
Let us consider a finite-dimensional $C^*$-algebra $\mathfrak{A}$ and a finite-dimensional faithful representation $(\mathcal{H},\pi)$, that is
\begin{equation}
\pi(A)=0 \; \Leftrightarrow \; A=0.
\end{equation}
Given a  state $\omega$ on $\mathfrak{A}$, it can be represented on $\pi(\mathfrak{A})$ by
\begin{equation}\label{dogsbark}
\omega_\pi\equiv\omega\circ \pi^{-1}.
\end{equation}
Let us decompose the representation into irreducible subrepresentations
\begin{equation}\label{decomprepr}
\left(\mathcal{H},\pi\right)=\bigoplus_{i=1}^N \left(\mathcal{H}_i^{(m_i)},\pi_i^{(m_i)}\right).
\end{equation}
Here, $(\mathcal{H}_i,\pi_i)$ are irreducible 
subrepresentations. The multiplicity of the subrepresentation $\pi_i$  is $m_i$, and
\begin{equation}
\mathcal{H}_i^{(m_i)}=\underbrace{\mathcal{H}_i\oplus\mathcal{H}_i\oplus\dots\oplus\mathcal{H}_i}_{m_i},\qquad \pi_i^{(m_i)}=\underbrace{\pi_i\oplus\pi_i\oplus\dots\oplus\pi_i}_{m_i}.
\end{equation}
The elements of $\pi(\mathfrak{A})$ have the form
\begin{equation}
X=\underbrace{X_1 \oplus X_1\oplus\dotsc\oplus X_1}_{m_1}\oplus \underbrace{X_2\oplus X_2\oplus\dotsc\oplus X_2}_{m_2}\oplus\dotsc\oplus \underbrace{X_N\oplus X_N\oplus\dotsc\oplus X_N}_{m_N},
\end{equation}
with $X_i$ spanning all $\mathcal{B}(\mathcal{H}_i)$, by the structure theorem -- see equation~(\ref{structuttheoruno}).
 
From representation~(\ref{decomprepr}), we can obtain another, more economical faithful representation of the form
\begin{equation}\label{reprnodeg}
\left(\tilde{\mathcal{H}},\tilde{\pi}\right)=\bigoplus_{i=1}^N \left(\mathcal{H}_i,\pi_i\right)
\end{equation}
where the multiplicities are $m_i=1$ for all $i$, thus eliminating all the redundancy of our description. For the moment, we stick with the general form~(\ref{decomprepr}), but we clearly expect that our results will not depend on the multiplicity $m_i$. 

We rewrite the decomposition~(\ref{decomprepr}) in the form
\begin{equation}\label{alernativedecomp}
\left(\mathcal{H},\pi\right)=\bigoplus_{i=1}^N \left(\mathcal{H}_i\otimes\mathbb{C}^{m_i},\pi_i\otimes\mathbb{I}_{m_i}\right).
\end{equation}
This follows by considering the unitary transformation which acts on each $\mathcal{H}_i^{(m_i)}$ as
\begin{equation}
\xi_1\oplus\xi_2\oplus\dots\oplus\xi_{m_i} \in \mathcal{H}_i^{(m_i)}\longleftrightarrow \xi_1\otimes e_1+\xi_2\otimes e_2+\dots\xi_{m_i}\otimes e_{m_i}\in\mathcal{H}_i\otimes\mathbb{C}^{m_i},
\label{ftn:unitarytransf}
\end{equation}
where $\{e_1,e_2,\dots, e_{m_i}\}$ is an orthonormal basis of $\mathbb{C}^{(m_i)}$.

Given a state $\omega$ over the $C^*$-algebra $\mathfrak{A}$, by Theorem~\ref{reprstate} we can consider the unique representative density matrix $\rho_\omega$ belonging to $\pi(\mathfrak{A})$ such that
\begin{equation}
\omega(A)=\operatorname{Tr}(\rho_\omega \pi(A)).
\end{equation}
Since $\rho_\omega$ is an element of the algebra, it has the form
\begin{equation}\label{decompstate}
\rho_\omega=p_1 \left(\rho_1\otimes \frac{\mathbb{I}_{m_1}}{m_1}\right) \oplus p_2\left( \rho_2\otimes \frac{\mathbb{I}_{m_2}}{m_2}\right) \oplus \dots  \oplus p_N\left( \rho_N\otimes \frac{\mathbb{I}_{m_N}}{m_N}\right),
\end{equation}
where $\rho_i$ are density matrices of $\mathcal{B}(\mathcal{H}_i)$, and $\vec{p}=(p_1, \dots, p_N)$ is a probability vector.
Conversely, any density matrix of the form~(\ref{reprstate}) defines a state over $\mathfrak{A}$.

Given two states $\omega_a$ and $\omega_b$, and their representative density matrices $\rho_{a}$ and $\rho_{b}$, we have
\begin{equation}\label{convexcombinequiv}
\omega=\lambda\omega_a+(1-\lambda) \omega_b \; \Leftrightarrow \;  \rho_\omega=\lambda\rho_{a}+(1-\lambda)\rho_{b}.
\end{equation}
Therefore, a state $\omega$ is pure if and only if its density matrix is pure with respect to decompositions in density matrices of $\pi(\mathfrak{A})$. 

Let $\rho_\omega$ be  a pure state, and let~(\ref{decompstate}) be its decomposition. Then, we must have that all $\rho_i=0$, except for one $i$. For example, if $\rho_1,\rho_2$ were both different from zero, then we could decompose $\rho$ into two other density matrices of $\pi(\mathfrak{A})$. Thus a pure state $\rho_\omega$ has the form
\begin{equation}\label{purestatesalgebra}
\rho_\omega=0 \oplus\dots  \oplus\left( |{\psi^{(i)}}\rangle\langle{\psi^{(i)}}| \otimes \frac{\mathbb{I}_{m_i}}{m_i}\right)\oplus \dots  \oplus 0,
\end{equation}
for some $i$, with $\psi^{(i)}$ being a unit vector of $\mathcal{H}_i$.

Given a state $\omega$ over $\mathfrak{A}$, let its representative $\rho_\omega$ be in the form~(\ref{decompstate}). Consider the spectral decomposition of each density matrix~$\rho_i$,
\begin{equation}
\rho_i=\sum_j \lambda_j^{(i)} |{\psi_j^{(i)}}\rangle\langle{\psi_j^{(i)}}|,
\end{equation}
and obtain a decomposition of the density matrix $\rho_\omega$ into pure states
\begin{align}
\rho_\omega&=\bigoplus_{i=1}^{N}p_i \left( \rho_i \otimes \frac{\mathbb{I}_{m_i}}{m_i}\right)=\bigoplus_{i=1}^{N}p_i \left(\sum_{j} \lambda_j^{(i)} |{\psi_j^{(i)}}\rangle\langle{\psi_j^{(i)}}|\otimes \frac{\mathbb{I}_{m_i}}{m_i}\right)\nonumber\\
&=\bigoplus_{i=1}^{N}\sum_j p_i  \lambda_j^{(i)} \left(|{\psi_j^{(i)}}\rangle\langle{\psi_j^{(i)}}|\otimes \frac{\mathbb{I}_{m_i}}{m_i}\right)=\sum_{ij}p_i \lambda_j^{(i)}\rho_j^{(i)}\label{minimaldecomp}
\end{align}
with
\begin{equation}\label{orthogonaldecomp}
\rho_j^{(i)}=0 \oplus\dots  \oplus\left( |{\psi_j^{(i)}}\rangle\langle{\psi_j^{(i)}}| \otimes \frac{\mathbb{I}_{m_j}}{m_j}\right)\oplus \dots  \oplus 0.
\end{equation}
The weights of this decomposition are $p_i \lambda_j^{(i)}$. We shall see that this is the minimal decomposition, i.e. having the minimal Shannon entropy as in definition~(\ref{entropydef}), which will be then the entropy $\mathcal{S}(\omega)$ of the state $\omega$.

Consider a generic decomposition of $\rho_\omega$ into pure states
\begin{equation}\label{gendecomp}
\rho_\omega=\sum_{ij} w_j^{(i)} \sigma_j^{(i)},
\end{equation}
with $\sigma_j^{(i)}$:
\begin{equation}
\sigma_j^{(i)}=0 \oplus\dots  \oplus\left( |{\varphi_j^{(i)}}\rangle\langle{\varphi_j^{(i)}}| \otimes \frac{\mathbb{I}_{m_i}}{m_i}\right)\oplus \dots  \oplus 0.
\end{equation}
We gathered the pure states so that $\sigma_j^{(i)}$ has support  in $\mathcal{H}_i^{(m_i)}$. We also define
\begin{equation}
v_j^{(i)}=\frac{w_j^{(i)}}{p_i},\quad v_j^{(i)}\geqslant 0,\quad\sum_j v_j^{(i)}=1,
\end{equation}
so that $\rho_\omega$ has the canonical form~\eqref{decompstate}, where for all $i$ we have a decomposition of $\rho_i$ in vector states:
\begin{equation}
\rho_i=\sum_j  v_j^{(i)}|{\varphi_j^{(i)}}\rangle\langle{\varphi_j^{(i)}}|.
\end{equation}
The Shannon entropy of the decomposition~(\ref{gendecomp}) is
\begin{align}\label{shannonentropydecomp}
\mathrm{H}(\vec w)&=-\sum_{ij}p_i v_j^{(i)} \log (p_iv_j^{(i)})\nonumber \\
&=-\sum_{ij} p_i v_j^{(i)}\log p_i-\sum_{ij} p_i v_j^{(i)}\log v_j^{(i)}\nonumber \\
&= \mathrm{H}(\vec{p}) +\sum_{i}p_i\mathrm{H}(\vec{v}^{(i)})\nonumber \\
&\geqslant \mathrm{H}(\vec{p})+\sum_{i}p_i\mathcal{S}_{\mathrm{VN}}(\rho_i).
\end{align}
Here $\mathcal{S}_\mathrm{VN}(\rho_i)$ is the von Neumann entropy of the density matrix $\rho_i$, which, by Schr\"odinger's theorem, is always smaller than the Shannon entropy of any other decomposition of $\rho_i$.

Now, the last line of~(\ref{shannonentropydecomp}) is also the Shannon entropy of the decomposition~(\ref{minimaldecomp}). Therefore, the entropy~(\ref{entropydef}) reads
\begin{equation}\label{finalformula}
S(\omega)=\mathrm{H}(\vec{p})+\sum_{i}p_i\mathcal{S}_{\mathrm{VN}}(\rho_i).
\end{equation}
This is our main result, that expresses the entropy of a state $\omega$ over an algebra $\mathfrak{A}$ in terms of the canonical decomposition~\eqref{decompstate} of its distinguished representative density matrix $\rho_\omega$ belonging to a faithful representation~\eqref{alernativedecomp} of $\mathfrak{A}$. The entropy $S(\omega)$ is given by the sum of two contributions: the Shannon entropy $\mathrm{H}(\vec{p})$ of the probability vector $\vec{p}$  of the weights of the component density matrices $\rho_i$ in the irreducible subrepresentations plus the average von Neumann entropy of these components. Notice that, as expected, the result does not depend on the arbitrary multiplicities $m_i$ of the representation.

On the other hand, the von Neumann entropy of the distinguished representative density matrix $\rho_\omega$ in the representation~\eqref{alernativedecomp} in general differs from the entropy~\eqref{finalformula} of the state $\omega$:
\begin{align}
\mathcal{S}_{\mathrm{VN}}(\rho_\omega)&=\mathrm{H}(\vec{p}) + \sum_i p_i \mathcal{S}_{\mathrm{VN}}\left(\rho_i\otimes \frac{\mathbb{I}_{m_i}}{m_i}\right)\nonumber \\
&=\mathrm{H}(\vec{p}) + \sum_i p_i \left(\mathcal{S}_{\mathrm{VN}}(\rho_i)+ \log m_i\right)\nonumber \\
&= \mathcal{S}(\omega)+\sum_{i=1}^N p_i \log m_i.
\end{align}
Indeed, it contains an additional entropic term due to the redundancy of the representation, that is the presence of multiplicities $m_i$.

The equality between the two entropies is restored if one considers the most economical representation with no multiplicities~\eqref{reprnodeg}. In such a case the entropy of the state $\omega$ is equal to the von Neumann entropy of its distinguished representative density matrix $\rho_\omega$ and equality~\eqref{entropyalternative} holds.
This observation has a major consequence: since $S(\omega)$ is the von Neumann entropy of the representative density matrix of a representation with no multiplicities, it is a \textit{bona fide} entropy and possesses all the desired thermodynamic properties; in particular, by equation~(\ref{convexcombinequiv}), it is a concave function.

We have proved the following theorem which gathers our main results:
\begin{thm}[Entropy of a quantum state]\label{mainthm}
 Let  $\mathfrak{A}$ be a finite dimensional $C^*$-algebra. For any state $\omega$ over $\mathfrak{A}$ define its entropy as
 \begin{equation}
\mathcal{S}(\omega)=\inf{ \Big\{\mathrm{H}(\vec{p})\, : \, \omega=\sum_i p_i \omega_i, \;\,  \text{$\vec p$ probability vector, \; $\omega_i$ pure states} \Big\}}.
\end{equation}
Then $\omega\mapsto \mathcal{S}(\omega)$ is a nonnegative concave function which vanishes on pure states.

Moreover, let $\left(\mathcal{H},\pi\right)$ be a faithful finite-dimensional and multiplicity-free representation of $\mathfrak{A}$. Given a state $\omega$, let $\rho_\omega \in \pi(\mathcal{H})$ be the unique density matrix such that $\omega(A) = \operatorname{Tr}\bigl(\rho \pi(A)\bigr)$ for all $A\in\mathfrak{A}$.
 Then one has
\begin{equation}
\mathcal{S}(\omega)=\mathcal{S}_{\mathrm{VN}}(\rho_\omega),
\end{equation}
where $\mathcal{S}_{\mathrm{VN}}(\rho_\omega)= -\operatorname{Tr}(\rho \log \rho)$ is the von Neumann entropy of $\rho_\omega$.
\end{thm}

\subsection{Thermodynamic considerations}

In this section we will discuss the physical motivations of the definition~\eqref{entropydef} for the entropy of a quantum state $\omega$.
We will make use of thermodynamic considerations by extending to the algebraic framework von Neumann's beautiful argument, based on the notions of Einstein's gas and semipermeable walls~\cite{von1996mathematical,Peres2002}.
To this purpose, some preliminary considerations are necessary.

There is no immediate definition of eigenstates in the algebraic approach, and yet they are key ingredients in von Neumann's thermodynamic considerations. Instead, we can consider  states that have a \textit{definite value} for a given observable. If a state $\omega_a$ has a definite value for an observable $A$, every measurement of this observable will yield the same value $a$ on it. This can be expressed by saying that $\omega_a(A)=a$ and its variance is zero:
\begin{equation}\label{definitevalue}
\omega_a((A-a)^2)=0.
\end{equation}
Furthermore, we assume that this property is stable in the sense that if a second measurement of the same observable is performed just after the first, the same result is obtained. 

In the following we will consider the faithful representation $(\mathcal{H},\pi)$ of a finite dimensional $C^*$-algebra $\mathfrak{A}$, without multiplicities, as given by~(\ref{reprnodeg}), namely
\begin{equation}
(\mathcal{H},\pi)=\bigoplus_{i=1}^{N} (\mathcal{H}_i,\pi_i)
\label{eq:multfree}
\end{equation}
with $(\mathcal{H}_i,\pi_i)$ being irreducible 
sub-representations.
Consider an observable $A=A^*\in\mathfrak{A}$ and let $\pi(A)=\pi(A)^\dag$ be its representative.
Let $\left(\varphi_i\right)_i$ be its eigenstates with eigenvalues $\left(a_i\right)_i$ and suppose that $A$ (and thus $\pi(A)$) has nondegenerate spectrum, that is $a_i\neq a_j$ for $i\neq j$. 
Now, if 
the density matrix $\rho_a\in\pi(\mathfrak{A})$ is the representative of the state $\omega_a$ then $\rho_a=|{\varphi_j}\rangle \langle{\varphi_j}|$ for some $j$, and $a=a_j$. Indeed, equation~\eqref{definitevalue} reads
\begin{align}
\operatorname{Tr}[\rho_a(\pi(A)- a)^2]&=\operatorname{Tr}\left[\rho_a\left(\sum_i a_i |{\varphi_i}\rangle\langle{\varphi_i}|-a \sum_i |{\varphi_i}\rangle\langle{\varphi_i}|\right)^2\right]\nonumber \\
&=\sum_i \langle{\varphi_i}|{\rho_a\varphi_i}\rangle\left(a_i-a\right)^2=0.
\end{align}
Therefore $\rho_a$ has no support on  $\varphi_i$  whenever $a_i\neq a$. As a result, $a=\operatorname{Tr}[\rho \pi(A)]$ is an eigenvalue of $\pi(A)$, say $a=a_j$ for some $j$, and $\rho_a$ is supported on its eigenspace. Thus we have
\begin{equation}
\rho_a =
|{\varphi_j}\rangle\langle{\varphi_j}|. 
\end{equation}

We are now ready to apply von Neumann's argument. We have seen in the previous sections that by considering the faithful multiplicity-free representation~\eqref{eq:multfree} there is a one to one correspondence between states $\omega$ over $\mathfrak{A}$ and density matrices $\rho_\omega$ belonging to $\pi(\mathfrak{A})$, and pure states over $\mathfrak{A}$ correspond to vector states $\ket{\psi}\!\bra{\psi}$ belonging to $\pi(\mathfrak{A})$, which, by the above argument, are states with a definite value for a suitable nondegenerate observable.
Moreover, we have seen that the entropy of any state $\omega$ is equal to the von Neumann entropy of its distinguished representative $\rho_\omega$, as in equality~\eqref{entropyalternative}. Therefore, the strategy will be to use von Neumann's argument on the representation $\pi(\mathfrak{A})$.

Consider an ensemble of $M$ copies of a system prepared in a state $\omega$, represented by the density matrix $\rho\in \pi(\mathfrak{A})$. If $M$ is large enough, we expect the system to follow the laws of thermodynamics.
In order to obtain the entropy of the system, we need to evaluate the heat exchanged along a reversible transformation that brings the system from a reference state $\omega_0$, whose entropy $\mathcal{S}_0$ is assigned, to the state $\omega$. The entropy will be given by
\begin{equation}
\mathcal{S}_{\mathrm{gas}}=\mathcal{S}_0+\int_{\omega_0}^\omega\frac{dQ}{T}.
\end{equation}

In quantum mechanics, one chooses pure states as the reference states, and sets $\mathcal{S}_0=0$. In fact, it can be proved that pure states are isoentropic, and that two pure states can be connected adiabatically~\cite{von1996mathematical}. We are going to see that this is in general not true in the algebraic description, and that there are states that cannot be transformed into each other in this way.

Let us recall von Neumann's argument, which makes a clever use of a peculiar feature of quantum mechanics, later on named ``quantum Zeno effect''~\cite{misra1977zeno,facchi08}.
Consider two orthogonal vectors $\varphi$ and $\psi$ in $\mathcal{H}$. We explicitly construct the  adiabatic transformation from $\varphi$ to $\psi$. Fix an integer $k$, and define for $\nu=0,1,\dots,k$
\begin{equation}
\psi^{(\nu)}=\cos \left({\frac{\pi \nu}{2 k}}\right)\varphi+\sin\left({\frac{\pi \nu}{2k}}\right)\psi.
\label{eq:psinudef}
\end{equation}
with $\psi^{(0)}=\varphi$ and $\psi^{(k)}=\psi$.
Consider a family of non-degenerate self-adjoint operators $B^{(\nu)}$ such that $\psi^{(\nu)}$ is one of the possible eigenvectors. By measuring in sequence the observables corresponding to $B^{(1)},B^{(2)},\dots,B^{(k)}$ on the vector state $|{\varphi}\rangle\langle{\varphi}|$ one gets
\begin{equation}
|{\varphi}\rangle\langle{\varphi}|\xrightarrow{B^{(1)}}\rho^{(1)}\xrightarrow{B^{(2)}}\rho^{(2)}\xrightarrow{B^{(3)}}\dotsc\xrightarrow{B^{(k)}}\rho^{(k)}
\end{equation}
The fraction of states that goes from $\psi^{(\nu-1)}$ to $\psi^{(\nu)}$ in the measurement of $B^{(\nu)}$ is
\begin{equation}
\mathrm{P}(\psi^{(\nu-1)}\rightarrow\psi^{(\nu)})=|{\langle{\psi^{(\nu-1)}}|{\psi^{(\nu)}}\rangle}|^2=\cos^2\left(\frac{\pi}{2k}\right)
\end{equation}
and
\begin{equation}
\mathrm{P}(\varphi\rightarrow\psi)\geqslant\cos^{2k}\left(\frac{\pi}{2k}\right)\sim\left(1-\frac{\pi^2}{8k^2}\right)^{2k}\xrightarrow{k\rightarrow\infty} 1,
\end{equation}
so that for large $k$ we have a  transformation of $\varphi$ into $\psi$ with probability one.
Assuming that in the measurement no heat exchange occurs, we have:
\begin{equation}
\mathcal{S}(\ket{\psi}\!\bra{\psi})\geqslant \mathcal{S}(\ket{\varphi}\!\bra{\phi}).
\end{equation}
Since the transformation can be repeated in the opposite direction $\psi\rightarrow\varphi$, we get
\begin{equation}
\mathcal{S}(\ket{\psi}\!\bra{\psi})=\mathcal{S}(\ket{\varphi}\!\bra{\phi}).
\end{equation}
This proof works in quantum mechanics, where the algebra of observables is the full algebra $\mathcal{B}(\mathcal{H})$,
but has problems for a generic algebra $\mathfrak{A}$ subject to selection rules, whose representation $\pi(\mathfrak{A})$ is a proper subalgebra of $\mathcal{B}(\mathcal{H})$.

In order for the operator $B^{(\nu)}$ to be the representative of an observable, we need $|{\psi^{(\nu)}}\rangle\langle{\psi^{(\nu)}}|$ to be in $\pi(\mathfrak{A})$ for all $\nu$. Since pure states are vector states in a subspace $\mathcal{H}_i$ of~\eqref{eq:multfree}, $|{\psi^{(\nu)}}\rangle\langle{\psi^{(\nu)}}|$ are elements of $\pi(\mathfrak{A})$ if and only if the vectors $\psi$ and $\varphi$ in~\eqref{eq:psinudef} belong to the same Hilbert space $\mathcal{H}_i$. Only in this case we can prove that they are isentropic. Otherwise, they  cannot be transformed into each other by the procedure described above, and we cannot compare their entropies.
Physically, they represent pure states belonging to disjoint phases (or sectors) that cannot be connected by any physical operation.

We then call $s_1,s_2,\dots, s_N$ the entropies of the pure states whose representatives are in $\mathcal H_1,\mathcal{H}_2,\dots,\mathcal{H}_N$, respectively. From the entropy of pure states, we are going  to obtain the entropy of a generic mixed state.
We need to consider a reversible process that brings the ensemble to a final pure state. This is performed by introducing the concept of Einstein's gas: the copies of the quantum system are inserted into boxes $\mathcal{K}_i$ (a box for each copy), that are so thick and massive that the state of the system $\omega$ will not be affected by the motion of the boxes.
We then insert all these boxes into a larger box ${\mathcal{K}}$, that will be kept in contact with a \textit{reservoir} $\mathcal{R}$ at temperature $T$. The boxes will behave like a perfect gas if the temperature $T$ is high enough.

Consider the spectral decomposition  of the density matrix $\rho$ corresponding to the state $\omega$ in the representation $\pi$. We get that the decomposition
\begin{equation}
\rho=\sum_{i=1}^{N} \left(\sum_{j}p_i\lambda_j^{(i)} |{\psi_j^{(i)}}\rangle\langle{\psi_j^{(i)}}|\right),
 \qquad p_i,\lambda_j^{(i)}>0,\quad \sum_{ij}{p_i\lambda_j^{(i)}}=1,
\end{equation}
with $|{\psi^{(i)}_j}\rangle\langle{\psi_j^{(i)}}|\in \pi_i(\mathfrak{A})$,
corresponds to the decomposition into pure states of $\omega$,
\begin{equation}
\omega = \sum_{i,j} p_i\lambda_j^{(i)} \omega_{j}^{(i)},
\end{equation}
where the index $i$ labels different sectors. Define the non-degenerate self-adjoint operator
\begin{equation}
B=\sum_{i=1}^{N} \sum_{j} a_j^{(i)} |{\psi_j^{(i)}}\rangle\langle{\psi_j^{(i)}}|\in \pi(\mathfrak{A}), \quad \textnormal{with } a_j^{(i)}\neq a_k^{(\ell)} \textnormal{ for }i,j\neq \ell, k,
\end{equation}
representing the observable $A$, i.e.\ $B=\pi(A)$, and  for which $a_j^{(i)}$ are the possible outcomes of a measurement, and $\psi_j^{(i)}$ are the associated eigenvectors.

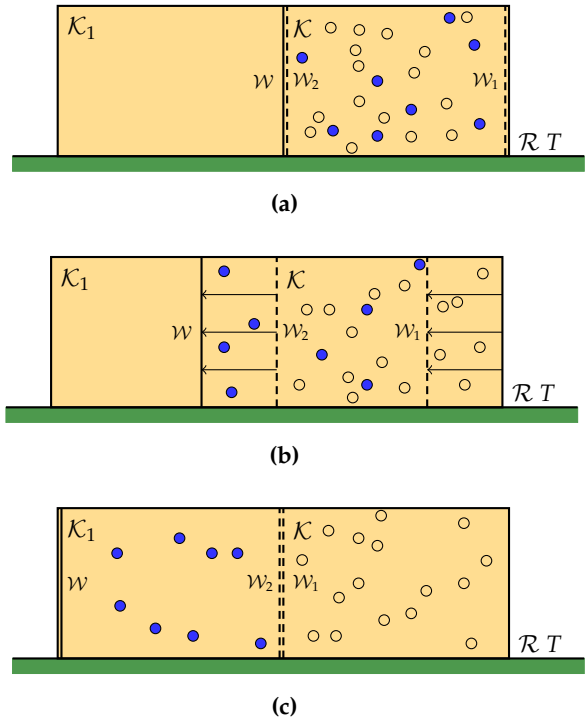
\begin{figure}[h!tb]
	\centering
	\begin{subfigure}{1\textwidth}\hspace{.27\textwidth}
			\begin{tikzpicture}
			\node[right] at (3,0.2) {\small $\mathcal{R}\ T$};
			\filldraw[greenlink] (-3.6,0)--(4,0)--(4,-0.2)--(-3.6,-0.2)--cycle;
			\draw[thick,fill=yellow!50!orange!50] (-3,2)--(-3,0)--(3,0)--(3,2)--cycle;
			\draw[thick,densely dashed] (2.95,0)--(2.95,2);
			\draw[thick,densely dashed] (0.05,0)--(0.05,2);
			\draw[thick] (0,0)--(0,2);
			\draw[thick] (-3.6,0)--(4,0);
			\node[right] at (0,1.7) {\small $\mathcal{K}$};
			\node[right] at (-3,1.7) {\small $\mathcal{K}_1$};
			\node[right] at (0,1) {\scriptsize $\mathcal{W}_2$};
			\node[left] at (3,1) {\scriptsize $\mathcal{W}_1$};
			\node[left] at (0,1) {\scriptsize $\mathcal{W}$};
			\draw[fill=blue!80] (2.54, 1.48) circle (2pt);
			\draw [fill=blue!80] (0.66, 0.34) circle (2pt);
			\draw (1.01, 1.68) circle (2pt);
			\draw (2.44, 1.85) circle (2pt);
			\draw (0.96, 1.41) circle (2pt);
			\draw (0.36, 0.32) circle (2pt);
			\draw (0.47, 0.52) circle (2pt);
			\draw (1,1.2) circle (2pt);
			\draw (1.88, 1.39) circle (2pt);
			\draw [fill=blue!80](1.25, 0.27) circle (2pt);
			\draw (0.63, 1.71) circle (2pt);
			\draw (1.01, 0.73) circle (2pt);
			\draw (2.24, 0.28) circle (2pt);
			\draw [fill=blue!80](2.21, 1.84) circle (2pt);
			\draw (1.74, 1.11) circle (2pt);
			\draw (1.34, 0.51) circle (2pt);
			\draw (0.91, 0.11) circle (2pt);
			\draw [fill=blue!80](1.7, 0.62) circle (2pt);
			\draw (2.17, 0.70) circle (2pt);
			\draw (1.38, 1.63) circle (2pt);
			\draw [fill=blue!80](2.61, .43) circle (2pt);
			\draw [fill=blue!80](0.25, 1.31) circle (2pt);
			\draw (1.7, 0.26) circle (2pt);
			\draw [fill=blue!80](1.25,1) circle (2pt);
			\end{tikzpicture}
			\caption{}
			\label{fig:figonea}
	\end{subfigure}%

	\begin{subfigure}{1\textwidth}
			\begin{tikzpicture}\hspace{.27\textwidth}
			\filldraw[greenlink] (-3.6,0)--(4,0)--(4,-0.2)--(-3.6,-0.2)--cycle;
			\draw[thick,fill=yellow!50!orange!50] (-3,2)--(-3,0)--(3,0)--(3,2)--cycle;
			\draw[thick,densely dashed] (0,0)--(0,2);
			\draw[thick] (-1,0)--(-1,2);
			\draw[thick] (-3.6,0)--(4,0);
			\draw[thick,densely dashed] (2,0)--(2,2);
			\node[right] at (0,1.7) {\small $\mathcal{K}$};
			\node[right] at (-3,1.7) {\small $\mathcal{K}_1$};
			\node[left] at (2.05,1) {\scriptsize $\mathcal{W}_1$};
			\node[right] at (-.05,1) {\scriptsize $\mathcal{W}_2$};
			\node[left] at (-1,1) {\scriptsize $\mathcal{W}$};
			\draw[->] (0,.5)--(-1,.5);
			\draw[->] (0,1)--(-1,1);
			\draw[->] (0,1.5)--(-1,1.5);
			\node[right] at (3,0.2) {\small $\mathcal{R}\ T$};
			\draw[->] (3,.5)--(2,.5);
			\draw[->] (3,1)--(2,1);
			\draw[->] (3,1.5)--(2,1.5);
			\draw (2.5, 0.3) circle (2pt);
			\draw (0.7, 1.3) circle (2pt);
			\draw (1, 1) circle (2pt);
			\draw (2.4, 1.4) circle (2pt);
			\draw (0.95, 0.4) circle (2pt);
			\draw (0.4, 1.3) circle (2pt);
			\draw [fill=blue!80](-.7, .8) circle (2pt);
			\draw (2.7,0.8) circle (2pt);
			\draw [fill=blue!80](1.9, 1.9) circle (2pt);
			\draw [fill=blue!80](1.2, 1.3) circle (2pt);
			\draw [fill=blue!80](0.6, 0.7) circle (2pt);
			\draw (1.01, .13) circle (2pt);
			\draw (2.74, 1.78) circle (2pt);
			\draw (2.21, 1.34) circle (2pt);
			\draw [fill=blue!80](-.7, 1.81) circle (2pt);
			\draw (1.3, 1.51) circle (2pt);
			\draw [fill=blue!80](-0.3, 1.11) circle (2pt);
			\draw (1.7, 1.62) circle (2pt);
			\draw (2.17, 0.70) circle (2pt);
			\draw (1.4, 0.6) circle (2pt);
			\draw [fill=blue!80](-.6, .2) circle (2pt);
			\draw (0.3, 0.3) circle (2pt);
			\draw (1.7, 0.26) circle (2pt);
			\draw [fill=blue!80](1.2,0.3) circle (2pt);
			\end{tikzpicture}
			\caption{}
			\label{fig:figoneb}
	\end{subfigure}%

	\begin{subfigure}{1\textwidth}\hspace{.27\textwidth}
			\begin{tikzpicture}
			\filldraw[greenlink] (-3.6,0)--(4,0)--(4,-0.2)--(-3.6,-0.2)--cycle;
			\draw[thick,fill=yellow!50!orange!50] (-3,2)--(-3,0)--(3,0)--(3,2)--cycle;
			\draw[thick,densely dashed] (0,0)--(0,2);
			\draw[thick,densely dashed] (-0.05,0)--(-0.05,2);
			\draw[thick] (-2.95,0)--(-2.95,2);
			\draw[thick] (-3.6,0)--(4,0);
			\node[left] at (0,1) {\scriptsize $\mathcal{W}_2$};
			\node[right] at (0,1) {\scriptsize $\mathcal{W}_1$};
			\node[right] at (-3,1) {\scriptsize $\mathcal{W}$};
			\node[right] at (3,0.2) {\small $\mathcal{R}\ T$};
			\node[right] at (0,1.7) {\small $\mathcal{K}$};
			\node[right] at (-3,1.7) {\small $\mathcal{K}_1$};
			\draw (2.5, 0.2) circle (2pt);
			\draw (0.7, .3) circle (2pt);
			\draw (1, 1.6) circle (2pt);
			\draw (2.4, 1) circle (2pt);
			\draw [fill=blue!80](-0.95, 1.4) circle (2pt);
			\draw (0.4, .3) circle (2pt);
			\draw (1.3, 1.9) circle (2pt);
			\draw (2.7,1.3) circle (2pt);
			\draw (1.9, .9) circle (2pt);
			\draw [fill=blue!80](-1.2, .3) circle (2pt);
			\draw (0.6, 1.7) circle (2pt);
			\draw (1.0, 1) circle (2pt);
			\draw (2.4, 1.8) circle (2pt);
			\draw [fill=blue!80](-2.21, 1.4) circle (2pt);
			\draw (.74,.81) circle (2pt);
			\draw (1.34,.51) circle (2pt);
			\draw [fill=blue!80](-0.3, .2) circle (2pt);
			\draw (1.7, .6) circle (2pt);
			\draw [fill=blue!80](-2.17, 0.7) circle (2pt);
			\draw [fill=blue!80](-1.38, 1.6) circle (2pt);
			\draw [fill=blue!80](-.61, 1.4) circle (2pt);
			\draw (0.25, 1.31) circle (2pt);
			\draw [fill=blue!80](-1.7, 0.4) circle (2pt);
			\draw (1.25,1.5) circle (2pt);
			\end{tikzpicture}
			\caption{}
			\label{fig:figonec}
	\end{subfigure}%
	\caption{On the left of the box $\mathcal{K}$ is placed another box $\mathcal{K}_1$, equal to it. Between them there are a wall $\mathcal{W}$ and a semipermeable wall $\mathcal{W}_2$, transparent only for the pure component $\omega_{j}^{(i)}$. On the right of the box $\mathcal{K}$ there is another semipermeable wall opaque only to the pure component $\omega_{j}^{(i)}$. If  $\mathcal{W}$ and $\mathcal{W}_1$ are translated to the left, by keeping their distance constant, the component $\omega_{j}^{(i)}$ is separated in a reversible way.}
	\label{fig:figone}
\end{figure}

To separate the pure components $\omega_j^{(i)}$ of the state $\omega$ represented by $|{\psi_j^{(i)}}\rangle\langle{\psi_j^{(i)}}|$, we use a \textit{semipermeable wall}, constructed as a wall with some windows on it. In particular, when a box $\mathcal{K}_i$ reaches a window, we let an engine open it and measure the observable $A$ on the state inside the box. If the result is a given value $a_j^{(i)}$, the engine lets the box pass; otherwise, it reflects it. In this way, the wall is transparent for the states $\omega_j^{(i)}$ and opaque for the others. Using such a wall, it is possible to separate the pure components (see figure~\ref{fig:figone}).

\begin{figure}[h!bpt]
	\centering
	\begin{subfigure}{.5\textwidth}\hspace{.2\textwidth}
			\begin{tikzpicture}
			\fill[yellow!50!orange!50](-2.5,2)--(-2.5,0)--(2.5,0)--(2.5,2)--cycle;
			\node[right] at (1.5,0.2) {\small $\mathcal{R}\ T$};   	 \filldraw[greenlink] (-2.5,0)--(2.5,0)--(2.5,-0.2)--(-2.5,-0.2)--cycle;   	 
			\draw[thick] (-1.5,2)--(-1.5,0)--(1.5,0)--(1.5,2)--cycle;
			\draw[thick] (-2.5,2)--(2.5,2);
			\draw[thick] (-2.5,0)--(2.5,0);
			\draw[thick] (0.5,0)--(0.5,2);
			\draw[->] (-1.5,.5)--(0.5,0.5);
			\draw[->] (-1.5,1)--(0.5,1);
			\draw[->] (-1.5,1.5)--(0.5,1.5);
			\node[right] at (-1.5,2.3) {$\mathcal{K}_i$};
			\node[left] at (-1.5,2.3) {$\mathcal{K}_{i-1}$};
			\node[right] at (1.5,2.3) {$\mathcal{K}_{i+1}$};
			\node[right] at (-1.5,1.75) {$\mathcal{V}$};
			\node[right] at (.5,1.75) {$\mathcal{V}_i$};
			\end{tikzpicture}
		\caption{}
		\label{fig:figtwoa}
	\end{subfigure}%
	\\
	\begin{subfigure}{.5\textwidth}\hspace{.2\textwidth}
			\begin{tikzpicture}
			\node[right] at (1.5,0.2) {\small $\mathcal{R}\ T$};
			\filldraw[greenlink] (-2.5,0)--(2.5,0)--(2.5,-0.2)--(-2.5,-0.2)--cycle;
			\draw[thick,fill=yellow!50!orange!50] (-1.5,2)--(-1.5,0)--(1.5,0)--(1.5,2)--cycle;
			\draw[thick] (-2.5,0)--(2.5,0);
			\draw[thick] (-1.,0)--(-1.,2);
			\draw[thick] (-0.6,0)--(-0.6,2);
			\draw[thick] (-0.3,0)--(-0.3,2);
			\draw[thick] (0,0)--(0,0);
			\draw[thick] (.2,0)--(0.2,2);
			\draw[thick] (0.5,0)--(0.5,2);
			\draw[thick] (0.7,0)--(0.7,2);
			\draw[thick] (0.9,0)--(0.9,2);
			\draw[thick] (1.,0)--(1,2);
			\draw[thick] (1.1,0)--(1.1,2);
			\draw[thick] (1.2,0)--(1.2,2);
			\draw[thick] (1.25,0)--(1.25,2);
			\draw[thick] (1.3,0)--(1.3,2);
			\draw[thick] (1.35,0)--(1.35,2);
			\draw[thick] (1.4,0)--(1.4,2);
			\draw[thick] (1.45,0)--(1.45,2);
			\node at (-1.35,2.3) {$\mathcal{K}_1$};
			\node at (-.7,2.3) {$\mathcal{K}_2$};
			\node at (-.1,2.3) {$\dotsc$};
			\end{tikzpicture}
		\caption{}
		\label{fig:figtwob}
	\end{subfigure}
	\caption{Each box is compressed reversibly in order to have the same density in all the boxes. The process is carried on isothermally at temperature $T$.}
	\label{fig:figtwo}
\end{figure}
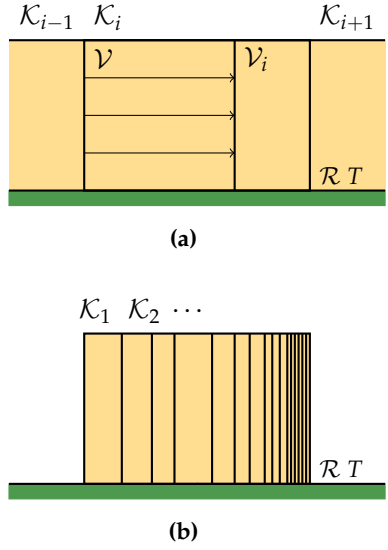

This process is reversible, and we get a final configuration of equal boxes, each containing one of the components $\omega_j^{(i)}$ of the gas.
We then compress each box isothermally, so that the system will have the same density of the original gas, see figure~(\ref{fig:figtwo}). The heat exchanged in each compression is given by
\begin{equation}
\mathrm{Q}=W=k_{\mathrm{B}} p_i\lambda_j^{(i)} M T \log{\frac{\mathcal{V}_{\mathrm{fin}}}{\mathcal{V}_{\mathrm{in}}}}=k_\mathrm{B} p_i\lambda_j^{(i)} MT\log\left({p_i\lambda_j^{(i)}}\right).
\end{equation}

The initial entropy of the gas is therefore
\begin{equation}
\mathcal{S}_{\mathrm{gas}}= -\sum_{ij}k_\mathrm{B} p_i\lambda_j^{(i)}  M\log\left({p_i\lambda_j^{(i)}}\right) + \mathcal{S}_{\mathrm{pure}}= M S_{\mathrm{VN}}(\rho) + \mathcal{S}_{\mathrm{pure}}.
\end{equation}
We now need to find  the entropy of the final configuration $\mathcal{S}_{\mathrm{pure}}$ consisting in separated pure components of the gas. Since entropy is an extensive quantity, it is given by the sum of the entropies of the pure components:
\begin{equation}
\mathcal{S}_\mathrm{pure}=\sum_{i=1}^{N}\left(\sum_j   p_i\lambda_j^{(i)} s_i  M\right)=\sum_{i=1}^N  p_i  s_i M.
\end{equation}
Therefore, we finally get
\begin{equation}
\frac{\mathcal{S}_{\mathrm{gas}}}{M} = \mathcal{S}_{\mathrm{VN}}(\rho)+\sum_i p_i s_i  = \mathcal{S}(\omega)+\sum_i p_i s_i ,
\label{eq:thermoent}
\end{equation}
where equality~\eqref{entropyalternative} was used. 

The entropy of the state $\omega$ obtained by thermodynamic considerations in~\eqref{eq:thermoent} differs from $\mathcal{S}(\omega)$ given in~\eqref{entropydef} by an additional term, $\sum_i p_i s_i$, which is the average of the arbitrary entropies $s_i$ assigned to pure states belonging to different phases.  By \emph{assuming} that pure states belonging to disjoint phases have the same entropy $s_1=s_2=\dots=s_N$ we get that the thermodynamic entropy is equal to the entropy $\mathcal{S}(\omega)$ up to an arbitrary constant, which we can set to~0. This is in agreement with the physical meaning of expression~\eqref{entropydef}, where the entropic content of a state $\omega$ is obtained exclusively as a result of  the mixing process with weights $p_i$ of pure states $\omega_i$ with zero entropy.

\subsection{Evaluation via the GNS construction}
In this last section we compute the entropy~\eqref{entropydef} of a quantum state $\omega$ by using the GNS representation of $\omega$. The problem of the ambiguity was studied in this framework by Balachandran, de Queiroz and Vaidya~\cite{balachandran2013entropy}. In particular, they described how to represent irreducible sub-representations as decomposition into pure states. This can be generalized for any decomposition. 

We start with the following result~\citep{bratteli2012operator,segal1947irreducible}.
\begin{thm}
Let $\omega$ be a state, and $(\mathcal{H}_\omega,\pi_\omega,\Omega_\omega)$ be its GNS representation. Then the following conditions are equivalent.
\begin{itemize}
\item $(\mathcal{H}_\omega,\pi_\omega)$ is irreducible;
\item $\omega$ is pure.
\end{itemize}
Moreover, there is a one to one relation between positive functionals $\lambda\omega_T$ over $\mathfrak{A}$ and majorized by $\omega$ and positive operators $T$ on $\mathcal{H}_\omega$ in the commutant $\pi'_\omega(\mathfrak{A})$ and with norm $\lVert{T}\rVert\leqslant 1$:
\begin{equation}
\lambda\omega_T(A)=\langle{\Omega_\omega}|{T\pi_\omega(A)\Omega_\omega}\rangle.
\end{equation}
\end{thm}

Notice that here, $\lambda$ is introduced in order to make $\omega_T$ a state. Moreover, we say that $\lambda\omega_T$ is majorized by $\omega$ if $\omega-\lambda\omega_T$ is positive, that is:
\begin{equation}
\omega(A^*A)-\lambda\omega_T(A^*A)\geqslant 0
\end{equation}

for all $A$. Observe that $\omega$ majorizes $\lambda\omega_T$ if and only if $\omega= \lambda\omega_T + (1-\lambda)\omega_S$ for some state $\omega_S$. Therefore, the above theorem links a convex decomposition to operators on a Hilbert space. In particular, one can prove that $\omega_T$ is pure if and only if $T$ is proportional to a projection $P_T$ in the commutant, and the corresponding sub-representation $(\mathcal{H}_T,\pi_T)$ is irreducible~\cite{bratteli2012operator}.

As a result, given a state $\omega$, it is equivalent to consider a decomposition into pure states $\omega_i$,
\begin{equation}
\omega=\sum_i \lambda_i \omega_i,
\end{equation}
or a decomposition of the identity of the representation in projections $P_i$,
\begin{equation}\label{decompoperator}
\mathbb{I}_{\mathcal{H}_\omega}=\sum_i t_i P_i,
\end{equation}
with $t_i\leqslant 1$ and
\begin{equation}\label{projectiondef}
\lambda_i\omega_i(A)=t_i\langle{\Omega_\omega}|{P_i\pi(A)\Omega_\omega}\rangle.
\end{equation}
The weights of the decomposition are obtained by evaluating equation~(\ref{projectiondef}) at $A=\mathbb{I}$:
\begin{equation}\label{weightsvalue}
\lambda_i=t_i\langle{\Omega_\omega}|{P_i\Omega_\omega}\rangle.
\end{equation}

Note that if $t_i=1$ for all $i$, the projections will be orthogonal to each other, and we obtain a decomposition of the GNS representation into irreducible sub-representations,
\begin{equation}\label{decomprepresentirr}
(\mathcal{H}_\omega,\pi_\omega)=\bigoplus_i(\mathcal{H}_i,\omega_i).
\end{equation}
This is the description given in~\cite{balachandran2013entropy}.

In the finite dimensional case, a decomposition into  irreducible sub-representations always exists, as well as a decomposition into pure states is always possible in a convex set (by Minkowski's theorem).
We can decompose the representation as
\begin{equation}
\left(\mathcal{H}_\omega,\pi_\omega\right)=\bigoplus_{i=1}^N \left(\mathcal{H}_i^{(m_i)},\pi_i^{(m_i)}\right)\equiv\bigoplus_{i=1}^N \left(\mathcal{H}_i\otimes\mathbb{C}^{m_i},\pi_i\otimes\mathbb{I}_{m_i}\right)
\end{equation}
using the unitary transformation~\eqref{ftn:unitarytransf}. By the structure theorem, the representation of the algebra is
\begin{equation}\label{algebrastructure}
\pi(\mathfrak{A})=\Big(\mathcal{B}(\mathcal{H}_1)\otimes\mathbb I_{m_1}\Big)\oplus\Big(\mathcal{B}(\mathcal{H}_2)\otimes\mathbb I_{m_2}\Big)\oplus\dots\oplus\Big(\mathcal{B}(\mathcal{H}_N)\otimes\mathbb I_{m_N}\Big)
\end{equation}
and its commutant is
\begin{equation}\label{commutant}
\pi'(\mathfrak{A})=\Big(\mathbb{I}_{\mathcal{H}_1}\otimes M_{m_1}\Big)\oplus\Big(\mathbb{I}_{\mathcal{H}_2}\otimes M_{m_2}\Big)\oplus\dots\oplus\Big(\mathbb{I}_{\mathcal{H}_N}\otimes M_{m_N}\Big).
\end{equation}
Thus, from~(\ref{commutant}), the irreducible projections have the form
\begin{equation}\label{irrproj}
P=\mathbb{I}_{\mathcal{H}_i}\otimes |{v}\rangle\langle{v}|
\end{equation}
for some $i$, with $v$ a unit vector in $\mathbb{C}^{m_i}$.

Therefore, given a family of irreducible projections $(P_j^{(i)})$, equation~(\ref{decompoperator}) becomes
\begin{equation}\label{problem}
\mathbb{I}_{\mathcal H_\omega}=\sum_{ij} t_j^{(i)} P_j^{(i)},
\end{equation}
with $t_j^{(i)}\leqslant 1$ and
\begin{equation}\label{minimalprojection}
P^{(i)}_j=\mathbb{I}_{\mathcal{H}_i}\otimes |{v_j^{(i)}}\rangle\langle{v_j^{(i)}}|.
\end{equation}
In particular, the index $i$ labels the sub-representation $\mathcal H_i\otimes \mathbb{C}^{m_i}$ considered, while $j$ labels the different projections in it.
From~(\ref{problem}) we get, for all $i=1,\dots, N$,
\begin{equation}\label{identitydecomp}
\mathbb{I}_{m_i}=\sum_j t_j^{(i)} |{v_j^{(i)}}\rangle\langle{v_j^{(i)}}|=\sum_j |{u_j^{(i)}}\rangle\langle{u_j^{(i)}}|,
\end{equation}
with
\begin{equation}
u_j^{(i)}=\sqrt{t_j^{(i)}}v_j^{(i)}.
\end{equation}

Consider now the normalized projection of $\Omega_\omega$ on $\mathcal{H}_i\otimes\mathbb{C}^{m_i}$, namely
\begin{equation}\label{projectionvacuum}
\Omega_i=\frac{1}{\sqrt{p_i}}\left(\mathbb{I}_{\mathcal H_i}\otimes\mathbb{I}_{m_i}\right)\Omega_\omega,
\end{equation}
where $p_i=\lVert{\left(\mathbb{I}_{\mathcal H_i}\otimes\mathbb{I}_{m_i}\right)\Omega_\omega}\rVert^2$. 
By plugging~(\ref{minimalprojection}) and~(\ref{projectionvacuum}) into equation~(\ref{weightsvalue}) we get
\begin{align}
\lambda^{(i)}_j&=t_j^{(i)}\langle{\Omega_\omega}|{P_j^{(i)}\Omega_\omega}\rangle=t_j^{(i)}p_i\langle{\Omega_i}|{(\mathbb{I}_{\mathcal H_i}\otimes|{v_j^{(i)}}\rangle\langle{v_j^{(i)}}|)\Omega_i}\rangle\nonumber\\
&=p_i\langle{\Omega_i}|{(\mathbb{I}_{\mathcal H_i}\otimes|{u_j^{(i)}}\rangle\langle{u_j^{(i)}}|)\Omega_i}\rangle\nonumber\\ \label{decompositiongns}
&=p_i\langle{u_j^{(i)}}|{\operatorname{Tr}_{\mathcal{H}_i}\left(|{\Omega_i}\rangle\langle{\Omega_i}|\right)u_j^{(i)}}\rangle=p_i\langle{u_j^{(i)}}|{\sigma_i u_j^{(i)}}\rangle
\end{align}
with
\begin{equation}\label{reducedstateomega}
\sigma_i=\operatorname{Tr}_{\mathcal{H}_i}\left(|{\Omega_i}\rangle\langle{\Omega_i}|\right)
\end{equation}

In general the decomposition of the identity in equation~(\ref{identitydecomp}) will consist of $M_i\geqslant m_i$ elements. If $(e_j^{(i)})_{j=1,\dots,m_i}$ is an orthonormal  basis of $\mathbb{C}^{m_i}$, it can be written as
\begin{equation}\label{orthogonal}
\delta_{kh}=\sum_{j=1}^{M_{i}}\langle{e_k^{(i)}}|{u_j^{(i)}}\rangle\langle{u_j^{(i)}}|{e_h^{(i)}}\rangle.
\end{equation}
This is an orthonormal relation between $m_i$ vectors of length $M_i$. We can expand the Hilbert space adding $M_i-m_i$ vectors $e^{(i)}_{m_i+1},\dots,e^{(i)}_{M_i}$, and obtain a complete orthonormal system in equation~(\ref{orthogonal}). Therefore, by setting
\begin{equation}
\tilde u_j^{(i)}=\sum_{k=1}^{M_i} \langle{e_k^{(i)}}|{u_j^{(i)}}\rangle e_k^{(i)}
\end{equation}
we will also get complete orthonormal system in $\mathbb{C}^{M_i}$. The operators $\sigma_i$ are defined so that they vanish on $e^{(i)}_j$ for $j>m_i$.

We now evaluate the Shannon entropy of the weight $\lambda_j^{(i)}$ in~(\ref{decompositiongns}):
\begin{align}
\mathrm{H}(\vec\lambda)=&-\sum_{ij} \lambda_j^{(i)}\log\lambda_j^{(i)}\nonumber\\
=&-\sum_{ij} p_i\langle{\tilde u_j^{(i)}}|{\sigma_i\tilde u_j^{(i)}}\rangle\log p_i\nonumber\\
&-\sum_{ij}p_i\langle{\tilde u_j^{(i)}}|{\sigma_i\tilde u_j^{(i)}}\rangle \log\langle{\tilde u_j^{(i)}}|{\sigma_i\tilde u_j^{(i)}}\rangle
\end{align}
Since $\Omega_i$ is normalized, the term in the second line will become 
\begin{equation}
-\sum_i p_i \log p_i=\mathrm{H}(\vec p).
\end{equation}
The term in the last sum takes its minimal value when $\tilde u_j^{(i)}$ are the eigenvectors of the reduced density matrix $\sigma_i$, becoming its von Neumann entropy. We finally get
\begin{equation}
\mathrm{H}(\vec\lambda)\geqslant \mathrm{H}({\vec p})+\sum_{i=1}^{M_i} p_i \mathcal{S}_{\mathrm{VN}}(\sigma_i)=\mathcal S(\omega),
\end{equation}
where formula~\eqref{finalformula} was used.

It is clear that we have re-obtained by this approach the  results previously obtained by using a faithful representation. However, some properties of the entropy---concavity, for example---are somewhat hidden in this description. Nevertheless, the derivation via the GNS construction might prove itself to be useful if one would like to extend these results to the infinite-dimensional case.

\section{Conclusions}\label{sec:conc}
We have seen that the ambiguity in the definition of the quantum entropy of a state can be traced back to an ambiguity in the definition of a representative on a Hilbert space, as different density matrices can be physically equivalent for a $C^*$-algebra of observables.

We started by observing the property of the von Neumann entropy to be the minimum of the Shannon entropies of the decompositions into  pure states. This minimality property was assumed to define unambiguously an entropy on the convex set of states over a $C^*$-algebra, obtaining a concave entropy, that generalizes the von Neumann entropy. We find that the theory can always be represented in an Hilbert space in which it yields the von Neumann entropy of a  suitable density matrix.

We also observed that it is possible to obtain this entropy by using thermodynamic arguments. The main difference with respect to  quantum mechanics is that we have to assume pure states to be isoentropic. In particular, we found that a theory can have disjoint sectors, associated with nontrivial invariant subspaces, and pure states of different sectors cannot be connected by a physical  process.

An interesting open problem would be the extension of our results  to an \textit{infinite-dimensional} $C^*$-algebra of observables. Here, new phenomena arise as there are states which are not represented by a density matrix and, in general, one expects e.g.\ to have decompositions given by an integral ---with a suitable measure $\mu$--- over the set of pure states.

\acknowledgments{This work was partially supported by Istituto Nazionale di Fisica Nucleare (INFN) through the project ``QUANTUM'', and by the Italian National Group of Mathematical Physics (GNFM-INdAM).}

\reftitle{References}


\externalbibliography{yes}
\bibliography{mybibfile}

\end{document}